\documentclass[a4paper,UKenglish]{lipics-v2016}
\usepackage{microtype}%if unwanted, comment out or use option "draft"

\usepackage{mdframed}
\title{Clearing an Orthogonal Polygon Using Sliding Robots}
\author[1,2]{Mohammad Ghodsi\thanks{This author's work was partially supported by IPM under grant no. CS-1392-2-01}}
\author[3]{Salma Sadat Mahdavi}
\author[4] {Ali Narenji Sheshkalani}
\affil[1]{ Computer Engineering Department, Sharif University of Technology,  Tehran, Iran\\ %Institute for Research in Fundamental Sciences (IPM), Iran
  \texttt{ghodsi@sharif.ir}}
\affil[2]{ Institute for Research in Fundamental Sciences (IPM), Tehran, Iran}
\affil[3]{ Computer Engineering Department, Sharif University of Technology, Tehran, Iran\\
  \texttt{ss.mahdavi110@gmail.com}}
\affil[4]{School of Electrical and Computer Engineering, University of Tehran, Tehran, Iran\\
  \texttt{narenji@ut.ac.ir}}
  
\authorrunning{Mohammad Ghodsi, Salma Sadat Mahdavi and Ali Narenji Sheshkalani} %mandatory. First: Use abbreviated first/middle names. Second (only in severe cases): Use first author plus 'et. al.'

\Copyright{Mohammad Ghodsi, Salma Sadat Mahdavi and Ali Narenji Sheshkalani}%mandatory, please use full first names. LIPIcs license is "CC-BY";  http://creativecommons.org/licenses/by/3.0/

% mandatory: Please choose ACM 1998 classifications from http://www.acm.org/about/class/ccs98-html . E.g., cite as "F.1.1 Models of Computation". 
% mandatory: Please provide 1-5 keywords %% motion planning
\EventLongTitle{27th International Symposium on Algorithms and Computation}
\EventShortTitle{ISAAC 2016}
\EventAcronym{ISAAC}
\EventYear{2016}
\EventDate{December 12--14, 2016}
\EventLocation{Sydney, Australia}
\EventLogo{}
\SeriesVolume{27}
\begin{document}

\maketitle

\begin{abstract}
In a multi-robot system, a number of autonomous robots would sense, communicate, and decide to move within a given domain to achieve a common goal.
In this paper, we consider a new variant of the pursuit-evasion problem in which the robots (pursuers) each move back and forth along an orthogonal line segment inside a simple orthogonal polygon $P$.  A point $p$ can be covered by a sliding robot that moves along a line segment $s$, if there exists a point $q\in s$ such that $\overline{pq}$ is a line segment perpendicular to $s$. 
In the pursuit-evasion problem, a polygonal region is given and a robot called a pursuer tries to find some mobile targets called evaders. The goal of this problem is to design a motion strategy for the pursuer such that it can detect all the evaders. 
We assume that $P$ includes unpredictable, moving evaders that have unbounded speed. We propose a motion-planning algorithm for a group of sliding robots, assuming that they move along the pre-located line segments with a constant speed to detect all the evaders with unbounded speed.
\end{abstract}

\section{Introduction}
The mathematical study of the ``pursuit-evasion'' problem was first considered by Parson \cite{parsons1978pursuit}. After that, the watchman route problem was introduced as a variation of the art gallery problem, which consists of finding static evaders in a polygon. 
The visibility-based motion-planning problem was introduced in 1997 by Lavalle et al. \cite{lavalle1997finding}. The aim was to coordinate the motions of one or more robots (pursuers) that have omnidirectional vision sensors to enable them to eventually ``see'' an evader that is unpredictable, has an unknown initial position, and is capable of moving arbitrarily fast. The process of detecting all evaders is also known as clearing the polygon. The pursuit-evasion problem has a broad range of applications such as air traffic control, military strategy, and trajectory tracking \cite{lavalle1997finding}.

In 2011, Katz and Morgenstern introduced sliding camera guards for guarding orthogonal polygons  \cite{katz2011guarding}. We define the ``sliding robots'' to be the same as the sliding cameras, where the robot $r_i$ would travel back and forth along an axis-aligned segment $s$ inside an orthogonal polygon $P$. A point $p$ is seen by $s_i$ if there exists a point $q\in s_i$ such that $\overline{pq}$ is a line segment perpendicular  to $s_i$ and is completely inside $P$. 
The set of all points of $P$ that can be seen by $s_i$ is its sliding visibility polygon (see Fig.\ref{SC}). The point $p$ is seen by $r_i$ if $r_i=q$.

\begin{figure}[h]
\centering
\includegraphics[height=1.4in]{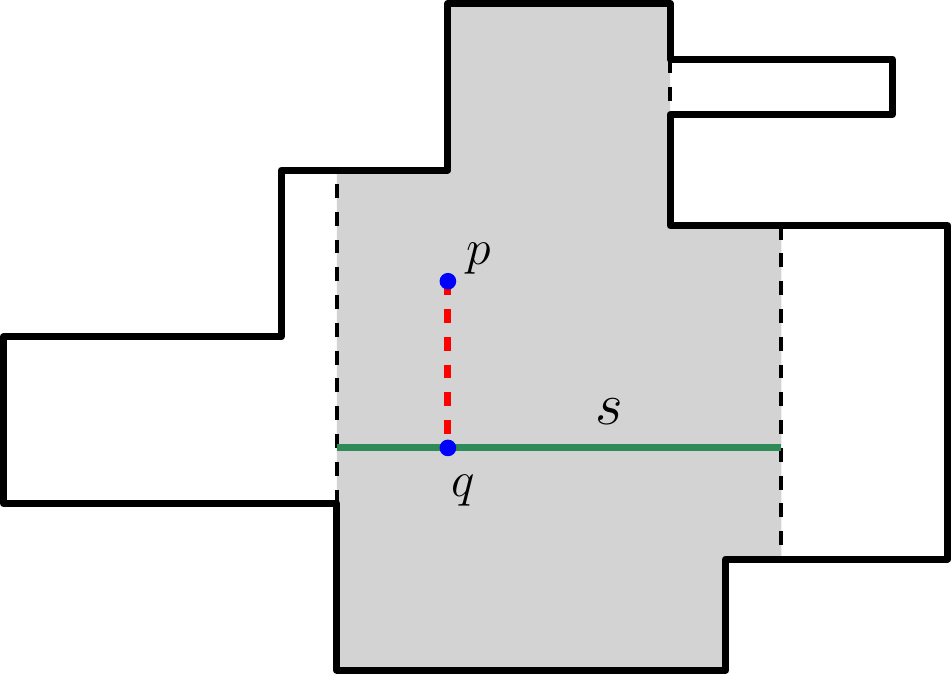}
 \caption{\small{The shaded area shows the sliding visibility polygon of $s$.}}
  \label{SC}
\end{figure} 
According to the visibility-based motion-planning problem and the sliding robots, we study the new version of planning the motions for a group of robots for clearing an orthogonal polygon when robots are modeled as sliding cameras. 
The given orthogonal polygon $P$ has unpredictable, moving evaders with unbounded speed. Motion planning for a group of sliding robots to clear $P$ means presenting the sequence of motions for the sliding robots such that any evader is viewed by at least one robot. Moreover, a set of line segments, $S$, is given such that the union of their sliding visibility polygons is $P$.

\section*{Previous Works}
Generally, in the pursuit-evasion problem, the pursuer is considered as an $l$-searcher with $l$ flashlights and rotates them continuously with a bounded angular rotation speed \cite{suzuki1992searching}. Thus, an $\infty$-searcher (also known as an omnidirectional searcher) is a mobile robot equipped with a $360^\circ$ view sensor for detecting evaders. Lavalle et al. proposed the first algorithm for solving the pursuit-evasion problem for an $l$-searcher \cite{lavalle1997finding}. They decomposed $P$ into cells based on visibility properties and converted the problem to a search on an exponential-sized information graph.
Durham et al. \cite{durham:2012} addressed the problem of coordinating a team of mobile robots with limited sensing and communication capabilities to detect any evaders in an unknown and multiply connected planar environment. They proposed an algorithm that guarantees the detection of evaders by maintaining a complete coverage of the frontier between cleared and contaminated regions while expanding the cleared region.

%-----------------------------
The art gallery problem is a classical and old problem in computational geometry. Over the years, many variants of this problem have been studied \cite{o1987art,urrutia2000art,hoffmann1990rectilinear,schuchardt1995two}. Most of these have been proved to be NP-hard \cite{lee1986computational}, containing the problem when the target region is a simple orthogonal polygon, and the goal is to find the minimum number of vertex guards to guard the entire polygon (e.g., \cite{o1987art,schuchardt1995two}). Some types of them, which consider the limited model of visibility, use polynomial time algorithms \cite{motwani1988covering,worman2007polygon}.

The study of the art gallery problem based on the sliding camera was started in 2011 by Katz and Morgenstern \cite{katz2011guarding}. They studied the problem of guarding a simple orthogonal polygon using minimum-cardinality sliding cameras (MCSC). They showed that, when the cameras are constrained to travel only vertically inside the polygon, the MCSC problem can be solved in polynomial time. They also presented a two-approximation algorithm for this problem when the trajectories that the cameras travel can be vertical or horizontal and the target region is an $x$-monotone orthogonal polygon. They left the computation of the complexity of the MCSC problem as an open problem.
In 2013, Durocher and Mehrabi \cite{durocher2013guarding} studied these two problems: the MCSC problem and the minimum-length sliding camera (MLSC) problem, where the goal was to minimize the total length of the trajectories along which the cameras travel. They proved that the MCSC problem is NP-hard, where the orthogonal polygon can have holes. They also proved that the MLSC problem is solvable in polynomial time even for orthogonal polygons with holes.
In 2014, Durocher {\it et al.} \cite{mehrabi20147} presented an $\mathcal{O}(n^{2.5})$-time $(7/2)$-approximation algorithm for solving the MCSC problem in simple orthogonal polygons.  In 2014, De Berg {\it et al.} \cite{de2014guarding} presented a linear-time algorithm for solving the MCSC problem in an $x$-monotone orthogonal polygon. The complexity of this problem remains as an open problem.
\section*{Our Result}

Our aim is to plan the motions for a group of robots that move along the line segments of $S$ and find all unpredictable evaders such that the number of robots used is the cardinality of $S$. Owing to the difficulty of having multiple cooperating robots executing common tasks, we store some information (e.g., the status of some nearby regions that shows whether the regions have been cleared by some robots) on each reflex vertex. 
\\
We assume that the robots have the map of the environment and that they are capable of broadcasting a message (e.g., a region that is supposed to get cleared) to all the other robots by sending signals. This way, the robots can have some communications with each other to maintain the coordination process. 
\\
The best result of our algorithm is that, if $S$ is a set of MCSCs that guard the whole $P$, then our algorithm will detect all evaders with the \textit{ \textbf{minimum number of sliding robots}}. 

%-----------------------------------------------------------------------------------------------------------------------------------------------------
\section{Preliminaries and Notations}
\label{PandN}
Let $P$ be an orthogonal polygon and $V(P)=\{v_1,v_2,..., v_n\}$ be the set of all vertices of $P$ in counterclockwise order. We consider $V_{ref}(P)$ to be all of the reflex vertices of $P$ and assume a general position such that no four reflex vertices are collinear. Suppose that $P_1$ is a sub-polygon of $P$ whose boundary is from $a$ to $b$ ($a$ and $b$ are points on the boundary of $P$) in counterclockwise order. Then, we show $P_1$ by $(a,b)$.

Let $v_j$ be a reflex vertex of $P$. $v_j$ has two edges, $e_{j-1}=\overline{v_{j-1},v_j}$ and $e_j=\overline{v_j,v_{j+1}}$, that can be extended inwardly until they reach the boundary of $P$. We call these extensions as the windows of $v_j$ and show them as $win_j(j-1,j)=\overline{v_j x_j}$ and $win_j(j+1,j)=\overline{v_j y_j}$, respectively ($x_j$ and $y_j$ are two points on the boundary).  $win_j(j-1,j)=\overline{v_j x_j}$ and $win_j(j+1,j)=\overline{v_j y_j}$ are two line segments whose endpoints are on the boundary of $P$. $win_j(j-1,j)$ partitions $P$ into two sub-polygons. 
Let $P_j(j-1,j)$ be a sub-polygon that consists of $v_{j+1}$, and let $P'_j(j-1,j)$ be $P\setminus P_j(j-1,j)$.  Therefore, $P_j(j-1,j)$ and $P'_j(j-1,j)$ are $(v_j,x_j)$ and $(x_j,v_j)$, respectively.

Similarly,  let $P_j(j+1,j)$ be a sub-polygon that is separated from $P$ by $win_j(j+1,j)$ and consists of $v_{j-1}$, and let $P'_j(j+1,j)$ be a sub-polygon that includes $v_{j+1}$. Therefore, $P_j(j+1,j)$ and $P'_j(j+1,j)$ are $(y_j,v_j)$ and $(v_j,y_j)$, respectively.  Let $L$ be the set of all windows of $P$. $L$ partitions $P$ into orthogonal rectangles. 

\begin{figure}[h]
\centering
\includegraphics[height=1.7in]{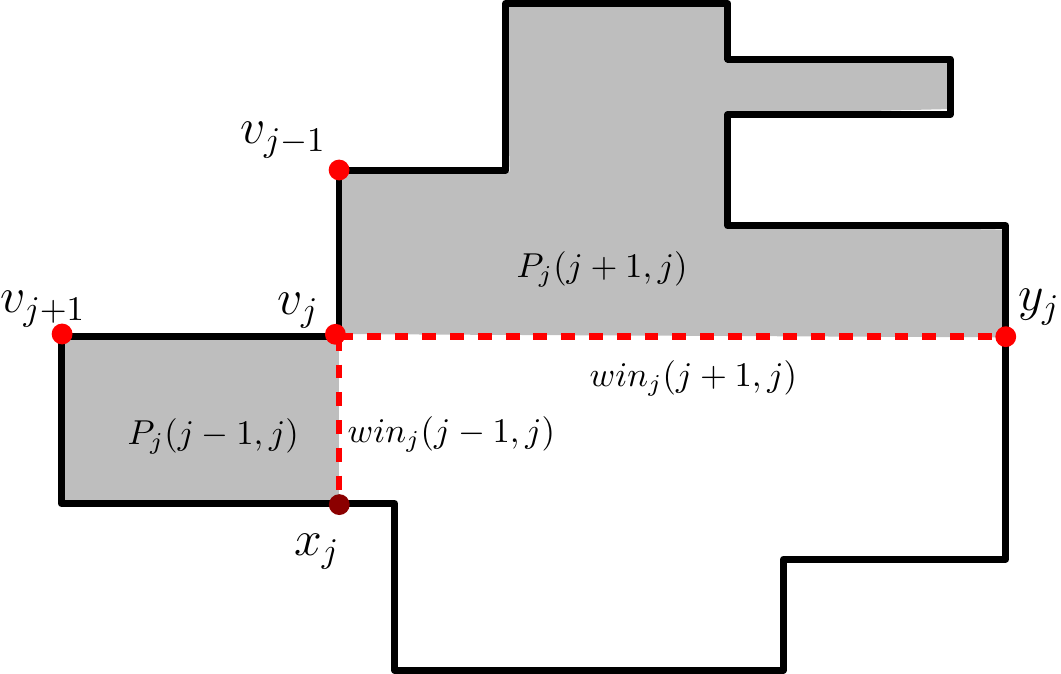}
 \caption{\small{Shown are the windows and the sub-polygons of $v_j$.}}
  \label{window}
\end{figure}
%---------------------------------------------------------------------------

%-----------------------------------------------------------------------------------------------------------------------------------------------------
\section{The Proposed Algorithm}
\label{algo}
In this section, we present an algorithm for solving the pursuit-evasion problem using sliding robots. Assume that an orthogonal polygon $P$ and a set of orthogonal line segments $S=\{s_1,s_2,\ldots ,s_k\}$ are given. We present a path-planning algorithm for finding the unpredictable evaders using a set of sliding robots $R=\{r_1,r_2,\ldots ,r_k\}$ in which $r_i$ can move along the line segment $s_i$. 
%We assume that the robots have the map of the environment, and they are capable of message broadcasting  through sending signals. This way, the robots can have communications with each other in order to maintain the coordination process. 
\\
%Additionally, as mentioned before, the status of the regions which are stored in the corresponding reflex vertices are updated by the robots during the movements to keep track of the contaminated regions.
To distribute the movements of the robots, we define the ``event points'' as below:
\begin{definition}
\label{eventpoint}
An \textbf{event point} happens when $r_i$ sees a reflex vertex, sees a waiting sliding robot, or reaches an endpoint of $s_i$.
\end{definition}
%We store the status of the regions in their corresponding reflex vertices which are updated by the robots during the movements to keep track of the contaminated regions and being helpful in decision making process.
\section*{Overview of the Algorithm}
Our algorithm has five steps. The ``start step,'' the ``decision step,'' the ``sending a signal,'' the ``move back-and-update step,'' and the ``termination step.''
%-------------------------------------
%We assume a partition into orthogonal rectangles for $P$,  %which is called $C(P)$
% such that each cell can be cleared continuously by at least one sliding robot. To do this purpose, we define $L_2$ and $L_3$ besides the definition of $L_1$ stated in Section \ref{notation}.
% 
%Let $L_2$ be the set of all maximal horizontal and vertical line segments which pass through the endpoints of vertical and horizontal line segments of $S$, respectively.  Also, let $L_3$ be the set of all maximal line segments which include line segments of $S$ and are contained in $P$ (in the other word, $L_3$ is extension of $S$).
%Let $L$ be a subdivision consists of the union of three subdivisions $L_1, L_2$ and $L_3$ (. \ref{}).
%---------------------------------
\\
To present our path-planning method, we start with an arbitrary sliding robot $r_i\in R$, which is on $s_i\in S$ (start step). $r_i$ starts moving from one endpoint of $s_i$. When $r_i$ reaches an event point, it updates the cleared sub-polygons. By the time that $D_i(2)$ becomes empty and $D_i(1)\neq \emptyset$, $r_i$ moves back along $s_i$ (move back-and-update step). Moreover, at each event point, $r_i$ stops and, according to the cleared sub-polygons of $P$, decides to continue its movement or send a signal to the other robots to clear a specific sub-polygon of $P$ (decision step).
When $r_i$ sends a signal to the other robots to clear a sub-polygon $P_1$, a robot that can clear some parts of $P_1$ starts moving along its corresponding line segment (sending a signal step). When all parts of $P$ become cleared, the algorithm is finished (termination step).
\section*{Details of the Algorithm}
Now, we explain the steps of the algorithm in detail. We store the status of the regions in their corresponding reflex vertices, which are updated by the robots during the movements to keep track of the contaminated regions, which is helpful in the decision-making process.

For each $v_j\in V_{ref}(P)$, we store an array called $FF_j(i)$ ($1\leq i \leq 4$) of size four in which the cells (of type Boolean) indicate whether the sub-polygons $P_j(j-1,j)$, $P_j(j+1,j)$, $P'_j(j-1,j)$, and $P'_j(j+1,j)$ are cleared (true), respectively. Initially, we assume that all parts of $P$ are contaminated;  therefore, $\forall 1\leq i\leq 4$, $FF_j(i)=false$. 
%$\big(P_j(j-1,j)$, $P_j(j+1,j)$, $P'_j(j-1,j)$ and $P'_j(j+1,j)$ are contaminated$\big)$.

For each $r_i\in R$, we consider a triple storage, which is called  $D_i(j)1 \leq j \leq 3$. Each storage includes an interval such as $(a,b)$, which indicates the boundary of $P$ between $a$ and $b$ in counterclockwise order.
The first storage, $D_i(1)$, indicates the cleared sub-polygon of $P$ by $r_i$ (partly or completely). The second storage, $D_i(2)$, indicates the sub-polygon of $P$ that should be cleared by $r_i$ (partly or completely). The third storage, $D_i(3)$, specifies the sub-polygon that should be cleared until $r_i$ continues its movement. In the case where $r_i$ is waiting, $D_i(3)$ is not empty. Initially, for each $r_i\in R$, $D_i(1)= D_i(2)=D_i(3)=\emptyset$.
%\begin{definition}
%We consider a triple storage for $r_i$ (called $D_i(j)1 \leq j \leq 3$), which the first cell, $D_i(1)$, indicates which sub-polygon of $P$ has been cleared (completely or partly) by $r_i$ (consists of two point $(a,b)$ on the boundary that means the boundary of $P$ which is between $a$ and $b$ in counterclockwise order is cleared by $r_i$). Similarly, the second cell, $D_i(2)$, indicates the sub-polygon of $P$ that should be cleared partly by $r_i$ (consists of two point $(c,d)$ on the boundary which means that the boundary of $P$ which is between them in counterclockwise order is going to be cleared completely or partly by $r_i$). The third cell specify that which sub-polygon should be cleared until $r_i$ continues it movement (consists of two point $(e,f)$ on the boundary which means that $r_i$ waits until the boundary of $P$ which is between $e$ and $f$ in counterclockwise order is cleared).
%\label{def2}
%\end{definition}
%-----------------------------------------------------------------------------------
%-------------------------------------------
%-----------------------------------------------------------------------------------
\\
\\
\textbf{Start Step}\\
As mentioned earlier, we start with one of the endpoints of an arbitrary $s_i$ ($r_i$ can move along $s_i$). 
\begin{itemize}
\item If $r_i$ starts from an endpoint that is on the boundary, $r_i$ can see two consecutive vertices (suppose the endpoint is on the edge $e_k=\overline{v_k v_{k+1}}$). 
	\begin{enumerate}%%[leftmargin=*]
		\item If $v_k$ and $v_{k+1}$ are convex, then $r_i$ starts clearing $P$ by its movement and updates $D_i(1)=(v_k, v_{k+1})$. $r_i$ continues its movement along $s_i$ until an event point happens. At these times, $r_i$ stops, updates $D_i(1)$ and $D_i(2)$, and makes a decision for its movement (decision step). 
		\item If at least one of $v_k$ or $v_{k+1}$ is a reflex vertex, then $r_i$ cannot start clearing $P$; it therefore stops and waits on the endpoint to make a decision (decision step).
	%\item If $r_i$ wants to start from the endpoint which is not on the boundary, then $r_i$ cannot start clearing $P$, stops and waits on the endpoint (Decision Step).
	\end{enumerate} 
\item If $r_i$  wants to start from an endpoint that is not on the boundary, then $r_i$ cannot start clearing $P$; it therefore stops and waits on the endpoint (decision step). Suppose that the maximal normal line segment to $s_i$ that passes through $r_i$ is $lr$. Let $x$ and $w$ be the first intersection of $lr$ at the boundary of two sides. $s_i$ can be inside the sub-polygon corresponding to $(x,w)$ or $(w,x)$. Assume that $s_i$ is inside $(w,x)$. Therefore, $r_i$ sends a signal to the other robots to clear $(x,w)$, and $D_i(3)=(x,w)$ (sending a signal step). As shown in Fig.\ref{upep}, if $r_i$ wants to start from $z$, it stops and sends a signal to the other robots to clear the sub-polygon corresponding to $(x,w)$.
\end{itemize}
%Also, $r_i$ stores the sub-polygon of $P$ which should be cleared until $r_i$ \
\begin{figure}[h]
\centering
\includegraphics[height=2.2in]{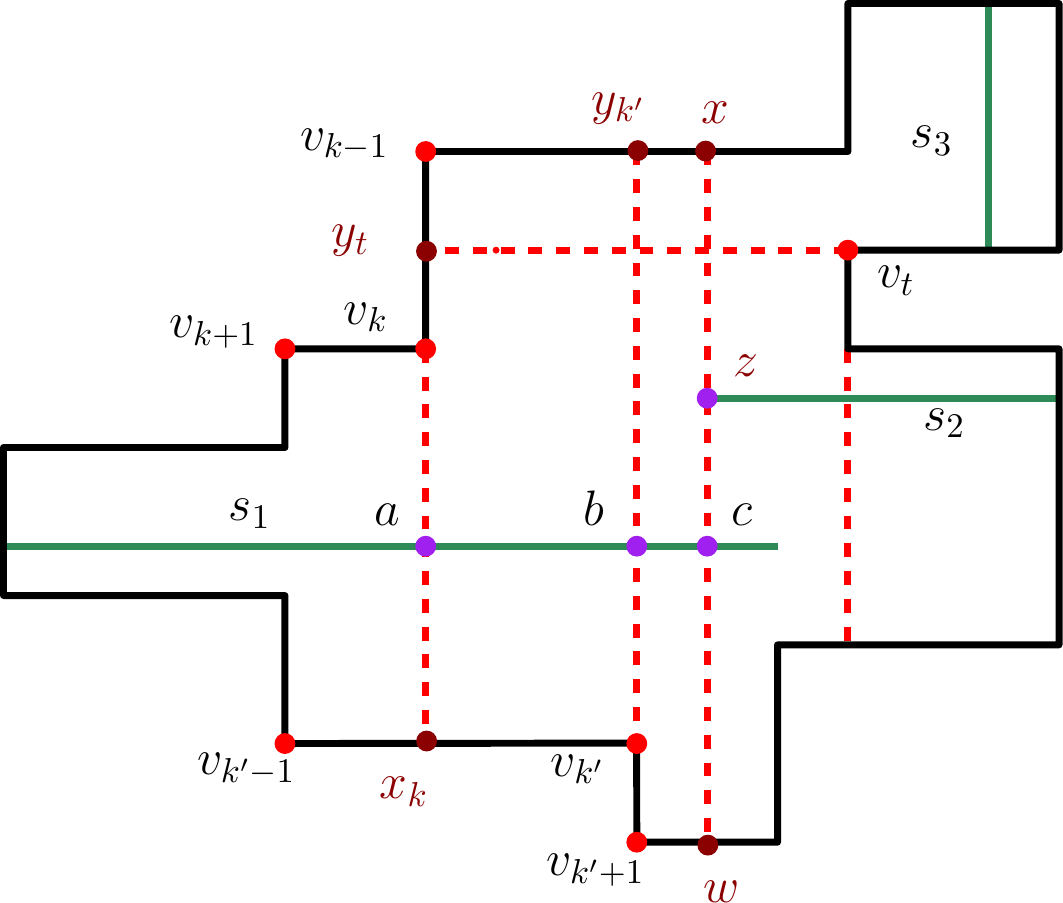}
 \caption{\small{$r_1, r_2$, and $r_3$ moving along $s_1, s_2$, and $s_3$, respectively.}}
  \label{upep}
\end{figure}
%-----------------------------------------------------------------------------------
\textbf{Move Back-and-Update Step}\\
Assume that $r_i$ moves along $s_i$. When an event point happens, %$r_i$ sees a reflex vertex, a waited robot  or reaches endpoint of $s_i$,
$r_i$ updates $D_i(1)$ (increases the cleared region) and $D_i(2)$ (decreases the sub-polygon that should be cleared). At each time that $D_i(2)$ becomes empty (and $D_i(1)\neq \emptyset$), $r_i$ moves back along $s_i$.  It moves back until it sees a waiting robot or reaches an endpoint of $s_i$. When $r_i$ sees a reflex vertex $v_k$ during its movement, it updates $FF_k(j)$ for $1\leq j \leq 4$ as detailed below:
%During its back movement when $r_i$ sees a reflex vertex $v_k$, it updates $FF_k(j)$ for $1\leq j \leq 4$ as below:

%Assume that the endpoints of  $win_k(k-1,k)$ are $v_k$ and $x$, and the endpoints of  $win_k(k+1,k)$ are $v_{k+1}$ and $y$.
\begin{itemize}%[leftmargin=0cm]
\item If the endpoints of $win_k(k-1,k)$ and $v_{k+1}$ are inside the sub-polygon indicated by $D_i(1)$  (if $D_i(1)=(v_k,x_k)$, then $v_{k+1}\in D_i(1)$), then the sub-polygon $P_k(k-1,k)$ is cleared and $r_i$ updates $FF_k(1)=true$ (see Fig.\ref{upep}, when $r_1$ moves back from left to right and reaches $a$).

\item If $D_i(1)=(x_k,v_k)$, then $r_i$ updates $FF_k(3)=true$ (see Fig.\ref{upep}, when $r_1$ moves back from right to left and reaches $a$).
%\item If the endpoints of $win_k(k-1,k)$ and $v_{k-1}$ are inside the sub-polygon indicated by $D_i(1)$ then the sub-polygon $P'_k(k-1,k)$ is cleared and $r_i$ updates $FF_k(3)=true$ (If $D_i(1)=(x_k,v_k)$ then $v_{k-1}\in D_i(1)$). See Fig \ref{upep}, when $r_1$ moves back from right to left and reaches $a$.

\item If $D_i(1)=(y_k,v_k)$, then  $r_i$ updates  $FF_k(2)=true$ (see Fig.\ref{upep}, when $r_1$ moves back from left to right and reaches $b$).

%\item If the endpoints of $win_k(k+1,k)$ and $v_{k-1}$ are inside the sub-polygon indicated by $D_i(1)$ then the sub-polygon $P_k(k+1,k)$ is cleared and $r_i$ updates  $FF_k(2)=true$ (If $D_i(1)=(y_k,v_k)$ then $v_{k-1}\in D_i(1)$). See Fig \ref{upep}, when $r_1$ moves back from left to right and reaches $b$.

\item If $D_i(1)=(v_k,y_k)$, then $r_i$ updates $FF_k(4)=true$ (see Fig.\ref{upep}, when $r_1$ moves back from right to left and reaches $b$).
%\end{prop}
\end{itemize}
%--------------------------
%\begin{center}
% \includegraphics[height=7cm]{figs/upep1.pdf}
% \label{upep}
%\end{center}

%--------------------------
As we explained earlier, $r_i$ moves back until it finishes clearing ($D_i(2)=\emptyset$). While it is moving back, if $r_i$ sees its corresponding waiting robot (supposedly $r_j$) and $D_i(1)=D_j(3)$, then $D_i(2)=\emptyset$. Therefore, $r_i$ updates $D_j(3)=\emptyset$, $D_j(1)=D_j(1)\bigcup D_i(1)$, and $D_j(2)=D_j(2)\bigcup D_i(1)$. Since $D_i(2)$ is empty, $r_i$ finishes its clearing and $r_j$ starts moving back (see Fig.\ref{upep}; when $r_1$ moves back from left to right and reaches $c$, it updates the information of $r_2$, and $r_2$ moves back).  $r_j$ can be collinear with the endpoint of $s_i$. Moreover, if $r_i$ sees any reflex vertex $v_k$, $r_i$ updates $FF_k(j)$ for $1\leq j \leq 4$ as explained above and continues moving back.
%(For example, See Fig \ref{upep}, when $r_3$ moves back from down to up on $s_3$ and reaches boundary).
%\begin{enumerate}%[leftmargin=*]
%	\item If $r_i$ sees its corresponding waited robot (suppose $r_j$) and $D_i(1)=D_j(3)$, then $r_i$ updates $D_j(3)=\emptyset$, $D_j(1)=D_j(1)\bigcup D_i(1)$ and $D_j(2)=D_j(2)\bigcup D_i(1)$. Thus, $r_i$ finish its movement and $r_j$ starts moving back (For example, See Fig \ref{upep}, when $r_1$ moves back from left to right and reaches $c$, update information of $r_2$ and $r_2$ moves back).  $r_j$ can be collinear with the endpoint of $s_i$.
%\label{1}
%	\item Else ($r_i$ does not see any waited robot), $r_i$ reaches its endpoints and its endpoint is on the boundary. If there is any reflex vertex $v_k$ collinear to the endpoint, then $r_i$ updates $FF_k(j)$ for $1\leq j \leq 4$ as explained above. Else (there is no reflex vertex collinear to $v_k$), $r_i$ finishes its clearing and the algorithm is finish (For example, See Fig \ref{upep}, when $r_3$ moves back from down to up on $s_3$ and reaches boundary).
%\end{enumerate}
%-----------------------------------------------------------------------------------
\\
\\
\textbf{Decision Step}\\
When $r_i$ stops and waits, it makes a decision and performs the following:
\begin{enumerate}%%[leftmargin=*]
\item If $r_i$ is on the endpoint of $s_i$ (let $ep$ be the endpoint), then 
%\item If $r_i$ reaches another endpoint of $s_i$ (suppose $ep$) then, 
	\begin{enumerate}%[leftmargin=0.2cm]
	\item If $ep$ is on the boundary of $P$ \big(on the edge $(e_k=\overline{v_k v_{k+1}})$\big), then
		\begin{itemize}%%[leftmargin=*]
		\item If $v_k\in V_{ref}(P)$ and $P_k(k+1,k)$ is contaminated ($FF_{k}(2)=false$), then  $P_k(k+1,k)$ should be cleared. Therefore, $r_i$ sends a signal to the other robots to clear $P_k(k+1,k)$ and updates $D_i(3)=(y_k,v_k)$ (As mentioned in Section \ref{PandN}, $y_k$ and $v_k$ are two endpoints of $win_k(k+1,k)$, and since $P_k(k+1,k)$ includes $v_{k-1}$, $D_i(3)$ is from $y_k$ until $v_k$ in counterclockwise order). \\For an example, see Fig.\ref{p1}; assume that $r_3$ or $r_2$ is on the blue point of $s_3$ and $s_2$, respectively. 
		%be two endpoints of $win_k(k+1,k)$ in counter clockwise order.
		\item Else if $P_k(k+1,k)$ is cleared ($FF_{k}(2)=true$), then $D_i(1)=D_i(1)\cup (y_k,v_k)$ and $D_i(2)=D_i(2) \setminus (y_k,v_k)$ .
		%\item If $v_k$ is a reflex vertex  and $P_k(k+1,k)$ is cleared ($FF_{k}(2)=true$), then $D_i(1)=D_i(1)\cup (y_k,v_k)$ and $D_i(2)=D_i(2) \setminus (y_k,v_k)$ .		

		%suppose that the endpoints of $win_k(k+1,k)$ are $v_k$ and $x$, then $D_i(1)=D_i(1)\cup (v_k,x)$ and $D_i(2)=D_i(2) \setminus (v_k,x)$ .		
		%-------------------------------------------------------------------------
		\item  If $v_{k+1}\in V_{ref}(P)$ and $P_{k+1}(k,k+1)$ is contaminated ($FF_{k+1}(1)=false$), then $P_{k+1}(k,k+1)$ should be cleared.  Therefore, $r_i$ sends a signal to the other robots to clear $P_{k+1}(k,k+1)$ and updates $D_i(3)=(v_{k+1},x_{k+1})$. \\For an example, see Fig.\ref{p1}; assume that $r_1$ or $r_2$ is on the blue point of $s_1$ and $s_2$, respectively. 
		\item  Else if $P_{k+1}(k,k+1)$ is cleared ($FF_{k+1}(1)=true$), then $D_i(1)=D_i(1)\cup (v_{k+1},x_{k+1})$ and $D_i(2)=D_i(2) \setminus (v_{k+1},x_{k+1})$.	
		%\item  If $v_{k+1}$ is a reflex vertex and $P_{k+1}(k,k+1)$ is cleared ($FF_{k+1}(1)=true$), then $D_i(1)=D_i(1)\cup (v_{k+1},x_{k+1})$ and $D_i(2)=D_i(2) \setminus (v_{k+1},x_{k+1})$ .	
		%suppose that the endpoints of $win_{k+1}(k,k+1)$ are $v_{k+1}$ and $y$, then $D_i(1)=D_i(1)\cup (v_{k+1},y)$ and $D_i(2)=D_i(2) \setminus (v_{k+1},y)$ .	

\begin{figure}[h]
\centering
\includegraphics[height=2.2in]{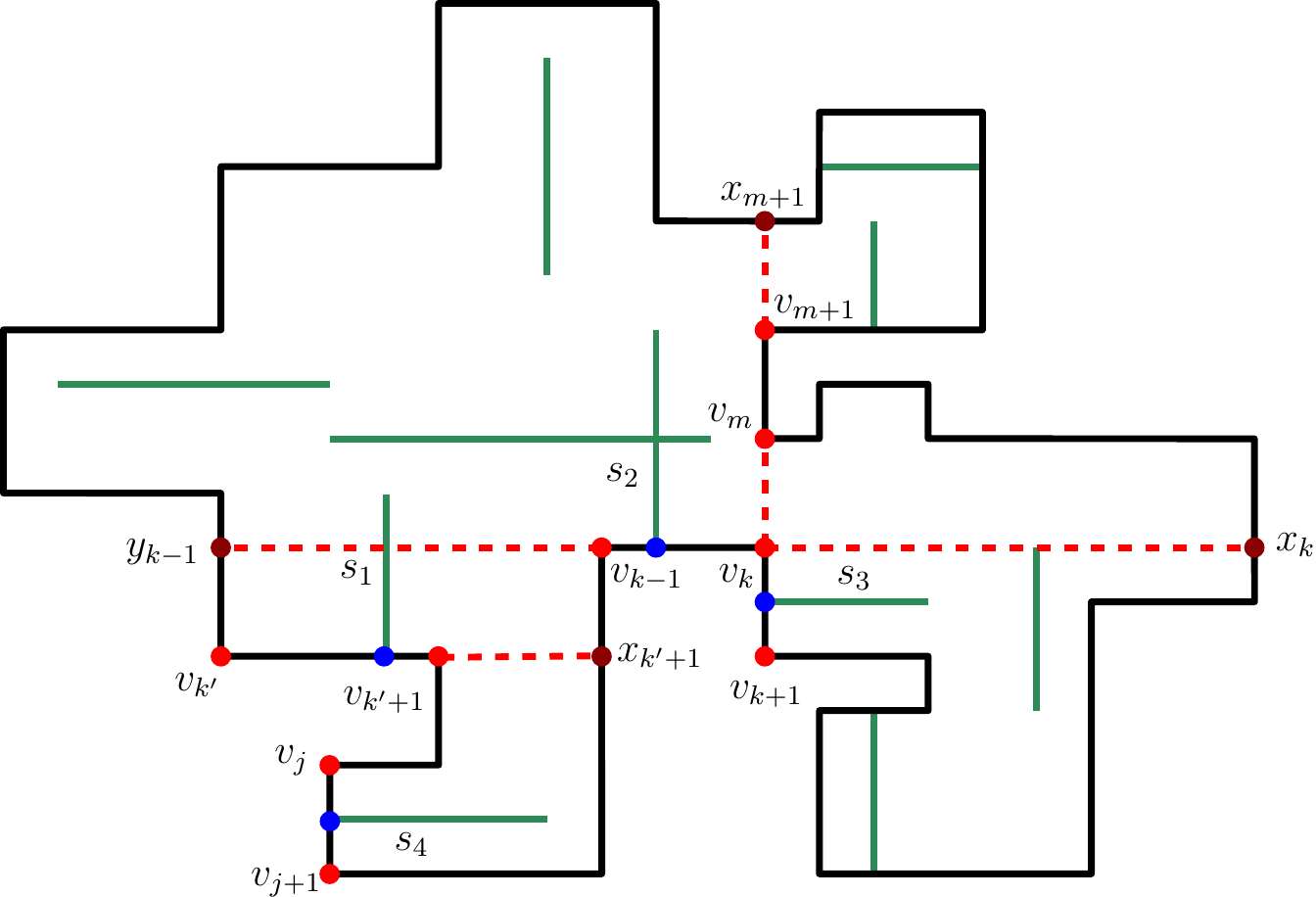}
 \caption{\small{$r_1, r_3$, and $r_4$ moving along $s_1, s_3$, and $s_4$, respectively}}
  \label{p1}
\end{figure}
		%-------------------------------------------------------------------------
     	 	  \item If at least one of $v_k$ and $v_{k+1}$ is a reflex vertex, then $ep$ is on $l(j)\in L$.  If $l(j)$ includes two consecutive reflex vertices $v_m, v_{m+1}$, where $m\neq k$ (suppose that the nearest one to $r_i$ is $v_m$), then \\For an example, See Fig.\ref{p1}; assume that $r_3$ is on the blue point of $s_3$. 
			  \begin{enumerate}%[leftmargin=0.4cm]
			  \item If $P_{m+1}(m,m+1)$ is contaminated ($FF_{m+1}(1)=false$), then $r_i$ sends a signal to the other robots to clear $P_{m+1}(m,m+1)$ and updates $D_i(3)=(v_{m+1},x_{m+1})$
			  % be two endpoints of $win_{m+1}(m,m+1)$ in counter clockwise order. 
			  \item Else ($FF_{m+1}(1)=true$), $D_i(1)=D_i(1)\cup (v_{m+1},x_{m+1})$ and $D_i(2)=D_i(2) \setminus (v_{m+1},x_{m+1})$ .	
			  %suppose that the endpoints of $win_{m+1}(m,m+1)$ are $v_{m+1}$ and $z$, then $D_i(1)=D_i(1)\cup (v_{m+1},z)$ and $D_i(2)=D_i(2) \setminus (v_{m+1},z)$ .	
		    	  \end{enumerate}
	        %-------------------------------------------------------------------------
		  \item If $v_k$ and $v_{k+1}$ are convex, then 
		  	\begin{enumerate}%%[leftmargin=*]
			\item If $r_i$ wants to start moving from $ep$ (if $D_i(1)=\emptyset$), then $r_i$ updates $D_i(1)=(v_k,v_{k+1})$ and starts moving along $s_i$.
			\item If $r_i$ reaches the endpoint of $s_i$ (if $D_i(1\neq \emptyset$), then $D_i(2)$ is $\emptyset$ and $r_i$ moves back.
			\end{enumerate}
		  \end{itemize}
               %-------------------------------------------------------------------------
	\item  If $ep$ is not on the boundary of $P$ and $ep$ is collinear by at least one reflex vertex, then $ep$ is on $l(j)\in L$. Therefore,
%	\item  $ep$ is not on the boundary of $P$, the last moving direction of $r_i$ is stored in $D(ep)$  

		\begin{itemize}%%[leftmargin=*]
%	\item If there is no two consequence reflex vertices collinear to $ep$ 
		\item If $l(j)$ consists of one reflex vertex $v_k$ (assume that the consecutive vertex of $v_k$ on $l(j)$ is $v_{k+1}$) and $s_i$ is inside $P_k(k+1,k)$, then \\For an example, see Fig.\ref{p4}; assume that $r_5$ is on the blue point of $s_5$. 

			%-------------------------------------------------------------------------
			\begin{enumerate}%%[leftmargin=*]
			\item If $P'_k(k+1,k)$ is contaminated ($FF_k(4)=false$), then $r_i$ sends a signal to the other robots to clear $P'_k(k+1,k)$ and updates $D_i(3)=(v_k,y_k)$.
			\item Else ($FF_k(4)=true$), $D_i(1)=D_i(1)\cup (v_k,y_k)$ and $D_i(2)=D_i(2) \setminus (v_k,y_k)$.
			% then $r_i$ does as describe in \ref{moveback}
			\end{enumerate} 
		\item Else if $s_i$ is inside $P'_k(k+1,k)$, then
			%-------------------------------------------------------------------------
   		        \begin{enumerate}%[leftmargin=*]
			\item If $P_k(k+1,k)$ is contaminated ($FF_k(2)=false$), then $r_i$ sends a signal to the other robots to clear $P_k(k+1,k)$ and updates $D_i(3)=(y_k,v_k)$.
			\item Else ($FF_k(2)=true$), $D_i(1)=D_i(1)\cup (y_k,v_k)$ and $D_i(2)=D_i(2) \setminus (y_k,v_k)$.
			\end{enumerate}
		%\item Else, if $l(j)$ consists of one reflex vertex $v_k$ (assume that the adjacent vertex of $v_k$ on $l(j)$ is $v_{k+1}$), and line segment $s_i$ is inside sub-polygon $F'(k+1,k)$ for vertex $v_k$ and $F(k+1,k)$ is contaminated ($FF_k(2)=false$), then $r_i$ send a signal for other robots to clear $F(k+1,k)$ and set $D_i(3)$ be two endpoints of $win_{k}(k+1,k)$ in counter clockwise order.
			%-------------------------------------------------------------------------
		\item If $l(j)$ consists of two consecutive reflex vertices $v_k$ and $v_{k+1}$ (suppose that the nearest one to $ep$ is $v_k$, and $s_i$ is inside $P_{k}(k+1,k)$, then \\For an example, see Fig.\ref{p4}; assume that $r_6$ is on the blue point of $s_6$.
%also suppose that the endpoints of $win_{k+1}(k,k+1)$ are $v_{k+1}$ and $y$ and the endpoints of $win_{k}(k+1,)$ are $v_k$ and $x$), 			
			\begin{enumerate}%[leftmargin=0.3cm]
			\item If $P_{k+1}(k,k+1)$ is contaminated ($FF_{k+1}(1)=false$), then  $r_i$ sends a signal to the other robots to clear it ($P_{k+1}(k,k+1)$) and updates $D_i(3)=(v_{k+1},x_{k+1})$.
			%\item Else, If $P_{k+1}(k,k+1)$ is cleared ($FF_{k+1}(1)=true$) and $P'_k(k+1,k)$ is contaminated ($FF_k(4)=false$) then  $r_i$ send a signal for the other robots to clear it and updates $D_i(3)=(x_{k+1},y_k)$.
			\item Else ($P_{k+1}(k,k+1)$ is cleared ($FF_{k+1}(1)=true$)) $r_i$ sends a signal to the other robots to clear $P'_{k+1}(k,k+1) \bigcap P'_k(k+1,k)$ and updates $D_i(3)=(x_{k+1},y_k)$.
%\item If $F'(k+1,k)$ for vertex $v_{k+1}$ is contaminated ($FF_{k+1}(4)=false$), then  $r_i$ send a signal for other robots to clear $F'(k+1,k)$ and set $D_i(3)=(y,v_{k+1})$.
			\end{enumerate}
			%-------------------------------------------------------------------------
		\item Else if $s_i$ is inside $P'_{k}(k+1,k)$, then
			\begin{enumerate}%[leftmargin=0.3cm]
			\item If $P_{k+1}(k,k+1)$ is contaminated ($FF_{k+1}(1)=false$), then $r_i$ sends a signal to the other robots to clear it ($P_{k+1}(k,k+1)$) and updates $D_i(3)=(v_{k+1},x_{k+1})$.
			\item If $P_{k}(k+1,k)$ is contaminated ($FF_{k}(2)=false$), then  $r_i$ sends a signal to the other robots to clear it and updates $D_i(3)=(y_k,v_k)$.
			\item If $P_{k+1}(k,k+1)$ and $P_{k}(k+1,k)$ are cleared ($FF_{k+1}(1)=true$ and $FF_{k}(2)=true$), then $D_i(1)=D_i(1)\cup (y_k,x_{k+1})$ and $D_i(2)= D_i(2)\setminus (y_k,x_{k+1})$.
			\end{enumerate}
	%\item If $l(j)$ consists of three visible reflex vertices $v_j$, $v_{j+1}$ and $v_k$: \\ Based on the nearest reflex vertex between $v_j$ and $v_{j+1}$ to $ep$ (suppose $v_{j+1}$), the next sub-polygon ($F(j+1,j)$) which is supposed to be cleared should be specified. If $F(j+1,j)$ is false, then the vertex's index $v_j$ is stored as $FF(j)$ in the vertex $v_j$.
		\end{itemize}
%-------------------------------------------------------------------------

\begin{figure}[h]
\centering
\includegraphics[height=2.2in]{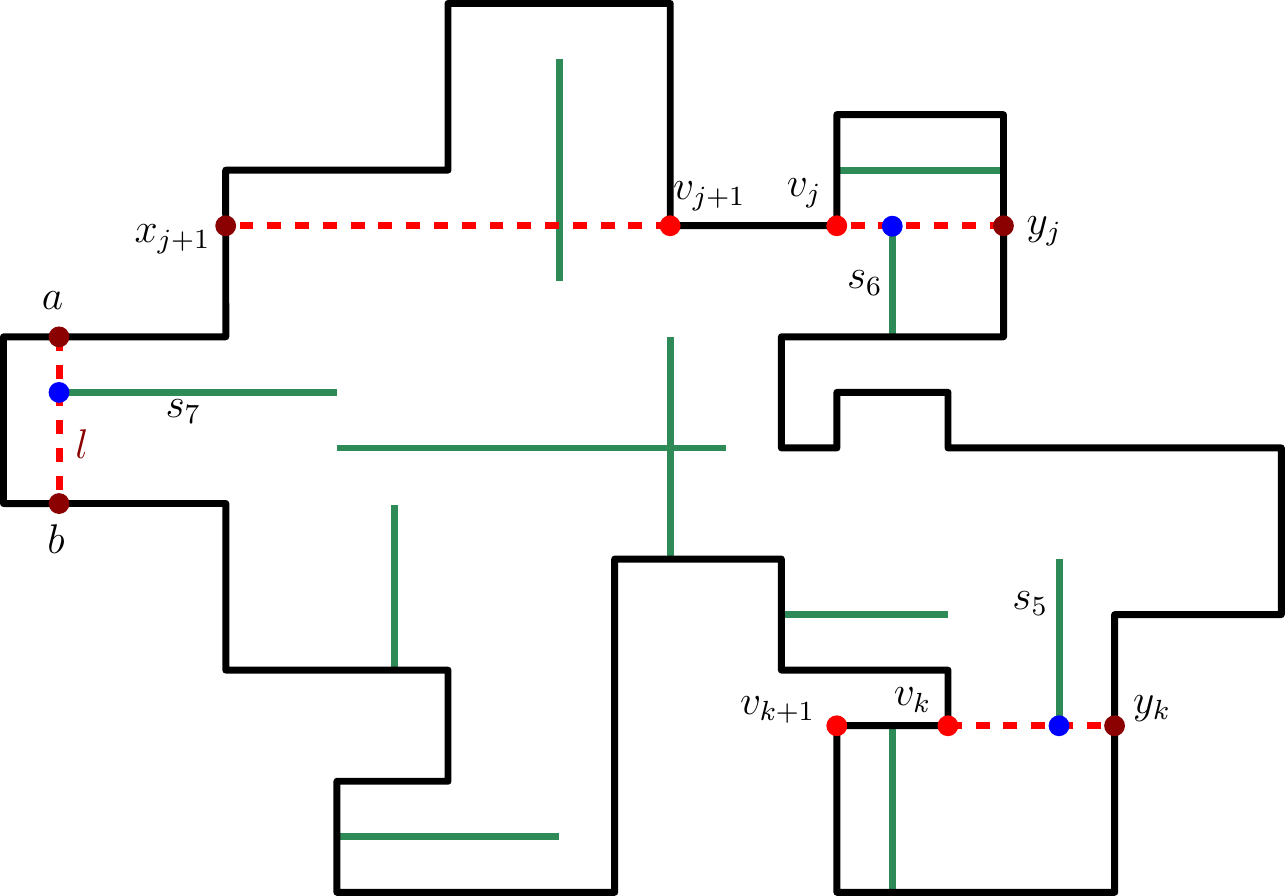}
 \caption{\small{$r_5, r_6$, and $r_7$ moving along $s_5, s_6$, and $s_7$, respectively}}
  \label{p4}
\end{figure}

%-------------------------------------------------------------------------
	\item  If $ep$ is not on the boundary of $P$ and $ep$ is not collinear by any reflex vertex, then suppose that the maximal orthogonal line segment normal to $s_i$ at $ep$ is $l$ and let $a$ and $b$ be two endpoints of $l$. $l$ partitions $P$ into two sub-polygons. One of them consists of $s_i$. Therefore, $r_i$ sends a signal to the other robots to clear the sub-polygon that does not include $s_i$ and that is between $a$ and $b$ $\big($ $r_i$ updates $D_i(3)$ depending on its position to $D_i(3)=(a,b)$ or $D_i(3)=(b,a)$$\big)$.  \\For an example, see Fig.\ref{p4}; if $r_7$ is on the blue point of $s_7$, then the sub-polygon that is between $(a,b)$ in counterclockwise order should be cleared.

	\end{enumerate}
%-------------------------------------------------------------------------
%-------------------------------------------------------------------------
\item Else if $r_i$ sees at least one reflex vertex ($r_i$ is on $l(j)\in L$), then 
	\begin{enumerate}%[leftmargin=*]
	\item If there are no two consecutive reflex vertices on $l(j)$, then $r_i$ continues its movement along $s_i$.
%	\item If $l(j)\in L_1$ and there are two consecutive reflex vertices $v_k, v_{k+1}$ on $l(j)$ (suppose that the nearest one to $r_i$ is $v_k$), then $r_i$ decides as below:
	\item If there are two consecutive reflex vertices $v_k, v_{k+1}$ on $l(j)$ (suppose that the nearest one to $r_i$ is $v_k$), then $r_i$ decides as below:  \\For an example, see Fig.\ref{p2}; assume that $r_1$ is on the point $p$  of $s_1$.
		\begin{itemize}%[leftmargin=0cm]
		\item If $P_{k+1}(k,k+1)$ is cleared ($FF_{k+1}(1)=true$), then $r_i$ updates $D_i(1)=D_i(1)\bigcup (v_{k+1},x_{k+1})$ and $D_i(2)= D_i(2)\setminus (v_{k+1},x_{k+1})$, and then continues its movement along $s_i$. 
	%is inside $SP^1_j$ and $SP^3_j$, and $SP^2_j$ is cleared, then $r_j$ continues its movement.
		\item If $P_{k+1}(k,k+1)$ is contaminated ($FF_{k+1}(1)=false$), then $P_{k+1}(k,k+1)$ should be cleared. Therefore, $r_i$ waits and sends a signal to the other robots to clear $P_{k+1}(k,k+1)$ and updates $D_i(3)=(v_{k+1},x_{k+1})$.
		%\item If $l(j)$ is inside $SP^1_j$ and $SP^3_j$, and $SP^2_j$ is contaminated, then $r_j$ waits on $l(j)$ until the other robot clears $SP^2_j$.\label{4b}
		%\item If $l(j)$ is inside $SP^2_j$ and $SP^3_j$, and $SP^1_j$ is contaminated, then $r_j$ waits on $l(j)$ until the other robot clears $SP^1_j$.\label{4d}
		\end{itemize}
	\end{enumerate}
%--------------------------------
%\end{enumerate}
%---------------------------
%-----------------------------------
%-----------------------------------

\end{enumerate}
\vspace{0.5cm}
%-----------------------------------------------------------------------------------
\textbf{Waiting and Sending a Signal Step}\\
Assume that $r_i$ waits and sends a signal to the other robots to clear sub-polygon $P_1$, which is between $a$ and $b$ in counterclockwise order ($D_i(3)=(a,b)$). 

When $r_i$ sends a signal, a robot that can clear some portions of $P_1$ consisting of $a$ starts clearing. 
At each time, one robot is clearing. Suppose that $r_j$ sees $a$ and can start clearing $P_1$. Therefore, $r_j$ updates $D_j(2)=D_i(3)$. 
\\
If $r_j$ is outside of $P_1$, $r_j$ starts clearing from $a$ and $D_j(1)$  is the intersection of the boundary of $P_1$ and the orthogonal line segment that passes through $a$ and intersects $s_j$. Therefore, $r_j$ starts its movement (see Fig.\ref{p2}). Otherwise ($r_j$ is inside $P_1$), $r_j$ starts clearing from one of its endpoints (for an example, see Fig.\ref{p4}; if $r_5$ is on the blue point of $s_5$,  $P_{k+1}(k,k+1)$ should be cleared). 
\begin{figure}[h]
\centering
\includegraphics[height=2in]{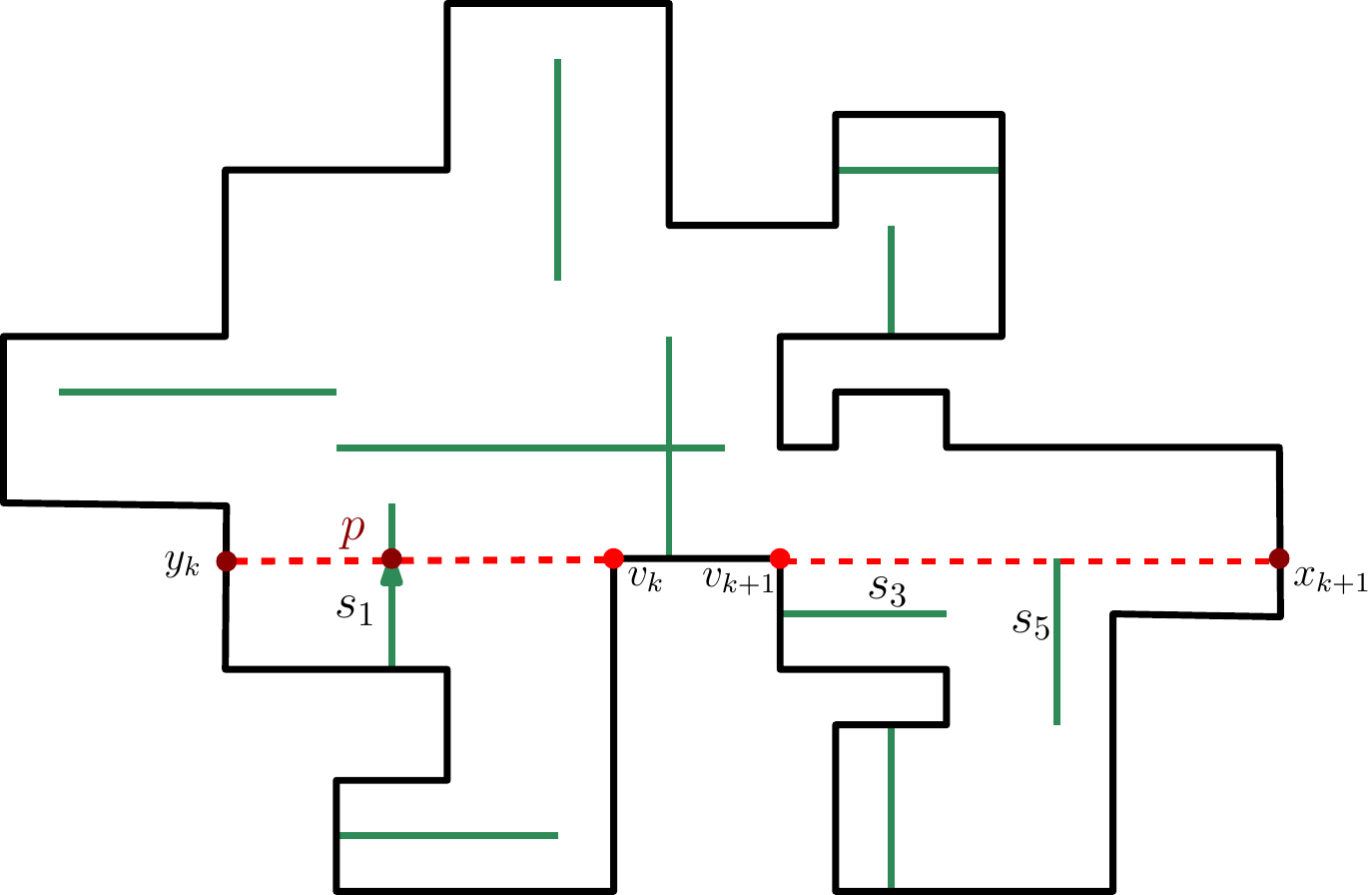}
 \caption{\small{When $r_1$ reaches $p$ and $FF_{k+1}(1)=false$, $r_3$ moves along $s_3$ from its left endpoint (or $r_5$ moves along $s_5$ from its upper endpoint).}}
  \label{p2}
\end{figure}
\\
Suppose that $v_k$ is a reflex vertex and that $FF_k(x)=false$ (let $P_1$ be the corresponding sub-polygon of $FF_k(x)$); suppose also that $r_j$ is a robot that is waiting until $P_1$ becomes cleared. At the time that $r_i$ updates $FF_k(x)$ to $true$, %(during the back movement),
$r_i$ finishes its clearance and updates $D_j(1)=D_j(1)\bigcup D_i(1)$ and $D_j(2)=D_j(2)\setminus D_i(1)$. Then, $r_j$ continues its movement.
%-------------------------------------------
\\
\\
%-----------------------------------------------------------------------------------
\textbf{ Termination Step Algorithm}\\
We assume that, initially, all parts of $P$ are contaminated and $\forall_{r_i\in R} D_i(1)=\emptyset$. Because of our algorithm, a robot can move and clear some parts of $P$ at any time. %Also, due to our general position assumption each robot $r_i$ can wait for clearance of at most two sub-polygons. When they are cleared, $r_i$ continues its movement. 
When there is no waiting robot ($\forall_{r_i\in R} D_i(3)=\emptyset$), all robots have cleared their corresponding sub-polygons ($\forall_{r_i\in R} D_i(2)=\emptyset$), and all parts of $P$ have been cleared ($ \bigcup_{i=1}^{|R|} D_i(1)=P$),  the motion-planning algorithm is finished.

%---------------------------------------------------------------------------
\section{Analysis}
In this section, we shall prove that the proposed algorithm is deadlock free. Since $S$ guards all parts of $P$, then the algorithm will be terminated. Then, we will show that, starting with any arbitrary sliding robot, the algorithm can clear $P$ completely.
\begin{lemma}
\label{lm:deadlock-free}
The proposed algorithm is deadlock free.
\end{lemma}
\begin{proof}
Assume that $r_i$ is waiting for sub-polygon $P_i$ to be cleared by a sequence of robots. Inside $P_i$, $r_j$ may be waiting for sub-polygon $P_j$ to be cleared. Therefore, there may exist a chain of waiting robots, say, $r_{seq}(i)=<r_j,r_t, \ldots , r_m>$, for clearing $P_i$. If $r_i \in r_{seq}(i)$, a deadlock occurs and the algorithm will not get terminated. Therefore, we shall show that the relation $r_i \in r_{seq}(i)$ will never become valid. \\
Owing to the definition of the window and its corresponding sub-polygons, when $r_i$ waits for the clearance of $P_i$, it cannot see any points of $P_i$, except its window. Since the sub-polygons corresponding to the other robots in $r_{seq}(i)$ are inside $P_i$, none of the waiting robots in $r_{seq}(i)$ can wait for $r_i$. Hence, the algorithm is deadlock free.

%Assume that $r_i$ is waiting for $r_j$ to clear the sub-polygon $P_q(v_{q+1}, v_q)$. As describe in \ref{algo}, $r_i$ cannot see any points of $P_q(v_{q+1}, v_q)$ except $win_q(q+1,q)$. As the sub-polygons corresponding to the other robots in $r_{seq}$ are inside $P_q(q+1,q)$, none of the waited robots in $r_{seq}$ waits for $r_i$. So, 

%\qed
%Assume that robot $r_i$ is waiting for $r_j$ to clear the corresponding region. Robot $r_j$ may wait for another robot in order to clear its region as well. Thus, there may exist a chain of waiting robots say $r_{seq}=<r_i,r_j, \ldots , r_k>$. If $r_i = r_k$, a deadlock is occurred, and tcannothe algorithm cannot get    terminated any more. Therefore, we show that the validity of relation $r_i = r_k$ is not possible to happen. 
%Assume that $r_i$ is waiting for $r_j$ to clear the sub-polygon $F(q, q+1)$ (Fig. \ref{}). Since the polygon $P$ is simple (there is no hole inside), $r_i$ is not visible to region $F(q, q+1)$ and the sub-polygons corresponding to the other robots in $r_{seq}$ which belong to $F(q, q+1)$ ($r_i$ is waiting until $F(q, q+1)$ gets cleaned by the robots in the sequence $r_{seq}$). 
%This way, $r_k$ can only be used to clear some portion of $F(q, q+1)$. Since $r_i$ is not able to see any portion of $F(q, q+1)$, none of the waiting robots in the sequence ask $r_i$ to clean a portion of $F(q, q+1)$.
%Hence, any loop of waiting robots ($r_{seq}$) may not get formed during the movement of the robots.
\end{proof}

\begin{lemma}
\label{lm:arbitraryStarting}
A simple orthogonal polygon can be completely cleared starting with an arbitrary sliding robot.
\end{lemma}
\begin{proof}
Assume that we start with an arbitrary robot $r_i$.  Because of Lemma \ref{lm:deadlock-free}, the proposed algorithm is deadlock free. Moreover, since $S$ guards all parts of $P$,  the termination step will happen. Based on the termination step, the relation $\bigcup_{i=1}^{|S|} D_i(1)=P$ becomes valid; therefore, there is no contaminated point in $P$ and the polygon gets cleared completely.

%As stated before, the starting sliding robot must be located on an endpoint of $S(P)$ which is laid on the boundary of $P$. Assume that robot $r_i$ starts to clear $P$ from an endpoint of the line segment $s_i \in S(P)$ (Fig. \ref{}). Based on the algorithm, the sliding robot move along the $s_i$ to clear the corresponding region until it reaches the other endpoint of $s_i$ (it may wait for other robots during the movement). Based on Lemma \ref{lm:deadlock-free}, $r_i$ can continue its movement and finally reaches the end of $s_i$. 
%%Since the union of visible regions of segments ($S(P)$) cover the whole $P$, whenever a robot wait for another one, 
%Assume that there exist a sub-polygon $P_{l}(m)$ where is needed to get cleared by some sliding robots. The sub-polygon $P_{l}(m)$ 
%
%It means that all the regions corresponding to reflex vertices which are seen by $r_i$ get cleared at the end of the movement.

%\qed
\end{proof}
%Due to Lemma \ref{lm:deadlock-free} and Lemma \ref{lm:arbitraryStarting}, we can conclude that:
\begin{theorem}
Let $P$ be a simple orthogonal polygon consisting of unpredictable evaders, and let $S$ be a set of line segments such that the union of their sliding visibility polygons is $P$. We can propose a motion-planning algorithm for a group of sliding robots that move along the line segments of $S$ and find all evaders such that the number of sliding robots used is at most the cardinality of $S$.
\end{theorem}

\begin{corollary}
If $S$ is the set of minimum cardinality sliding cameras that guard the whole $P$, then our algorithm clears $P$ with the minimum number of sliding robots.
\end{corollary}

%---------------------------------------------------------------------------
%\input{Conclusion.tex}
\section{Conclusion}
In this paper, we have proved that, in the case of having a known environment for sliding robots, there exists an algorithm for planning the motions of a group of sliding robots to detect all the unpredictable moving evaders that have unbounded speed. We assume that the speed of the sliding robots is unbounded ($\neq \infty$). We use a set of line segments $S$ where the sliding robots move along. In the case where $S$ is a set of minimum-cardinality sliding cameras that guard $P$, the proposed algorithm uses the minimum number of sliding robots to clear $P$. 

Investigating the problem in which the environment is unknown to the robots, and in which the robots could only plan their motions based on the local visible area, would be challenging. 
Additionally, letting the robots send information only to those that are visible to them may make the problem more usable in real-life multi-robot systems.

%---------------------------------------------------------------------------

%---------------------------- Bibliography -------------------------------

% Please add the contents of the .bbl file that you generate,  or add bibitem entries manually if you like.
% The entries should be in alphabetical order

\bibliographystyle{abbrv}
\bibliography{paper}

\begin{thebibliography}{10}

\bibitem{de2014guarding}
Mark de~Berg, Stephane Durocher, and Saeed Mehrabi.
\newblock Guarding monotone art galleries with sliding cameras in linear time.
\newblock In {\em Combinatorial Optimization and Applications}, pages 113--125.
  Springer, 2014.

\bibitem{durham:2012}
Joseph~W Durham, Antonio Franchi, and Francesco Bullo.
\newblock Distributed pursuit-evasion without mapping or global localization
  via local frontiers.
\newblock {\em Autonomous Robots}, 32(1):81--95, 2012.

\bibitem{durocher2013guarding}
Stephane Durocher and Saeed Mehrabi.
\newblock Guarding orthogonal art galleries using sliding cameras: algorithmic
  and hardness results.
\newblock In {\em Mathematical Foundations of Computer Science 2013}, pages
  314--324. Springer, 2013.

\bibitem{hoffmann1990rectilinear}
Frank Hoffmann.
\newblock {\em On the rectilinear art gallery problem}.
\newblock Springer, 1990.

\bibitem{katz2011guarding}
Matthew~J Katz and Gila Morgenstern.
\newblock Guarding orthogonal art galleries with sliding cameras.
\newblock {\em International Journal of Computational Geometry \&
  Applications}, 21(02):241--250, 2011.

\bibitem{lavalle1997finding}
Steven~M LaValle, David Lin, Leonidaa~J Guibas, Jean-Claude Latombe, and Rajeev
  Motwani.
\newblock Finding an unpredictable target in a workspace with obstacles.
\newblock In {\em Robotics and Automation, 1997. Proceedings., 1997 IEEE
  International Conference on}, volume~1, pages 737--742. IEEE, 1997.

\bibitem{lee1986computational}
Der-Tsai Lee and Arthur~K Lin.
\newblock Computational complexity of art gallery problems.
\newblock {\em Information Theory, IEEE Transactions on}, 32(2):276--282, 1986.

\bibitem{mehrabi20147}
Ali~D Mehrabi and Saeed Mehrabi.
\newblock A (7/2)-approximation algorithm for guarding orthogonal art galleries
  with sliding cameras.
\newblock In {\em LATIN 2014: Theoretical Informatics: 11th Latin American
  Symposium, Montevideo, Uruguay, March 31--April 4, 2014. Proceedings}, volume
  8392, page 294. Springer, 2014.

\bibitem{motwani1988covering}
Rajeev Motwani, Arvind Raghunathan, and Huzur Saran.
\newblock Covering orthogonal polygons with star polygons: The perfect graph
  approach.
\newblock In {\em Proceedings of the fourth annual symposium on Computational
  geometry}, pages 211--223. ACM, 1988.

\bibitem{o1987art}
Joseph O'rourke.
\newblock {\em Art gallery theorems and algorithms}, volume~57.
\newblock Oxford University Press Oxford, 1987.

\bibitem{parsons1978pursuit}
Torrence~D Parsons.
\newblock Pursuit-evasion in a graph.
\newblock In {\em Theory and applications of graphs}, pages 426--441. Springer,
  1978.

\bibitem{schuchardt1995two}
Dietmar Schuchardt and Hans-Dietrich Hecker.
\newblock Two np-hard art-gallery problems for ortho-polygons.
\newblock {\em Mathematical Logic Quarterly}, 41(2):261--267, 1995.

\bibitem{suzuki1992searching}
Ichiro Suzuki and Masafumi Yamashita.
\newblock Searching for a mobile intruder in a polygonal region.
\newblock {\em SIAM Journal on computing}, 21(5):863--888, 1992.

\bibitem{urrutia2000art}
Jorge Urrutia et~al.
\newblock Art gallery and illumination problems.
\newblock {\em Handbook of computational geometry}, 1(1):973--1027, 2000.

\bibitem{worman2007polygon}
Chris Worman and J~Mark Keil.
\newblock Polygon decomposition and the orthogonal art gallery problem.
\newblock {\em International Journal of Computational Geometry \&
  Applications}, 17(02):105--138, 2007.

\end{thebibliography}

\end{document}